\documentclass[a4paper,11pt]{JHEP3}

\usepackage{latexsym}
\usepackage{amsmath,amsthm}
\usepackage{amsfonts}
\usepackage{graphics}  
\usepackage{graphpap}  
\usepackage{amssymb, epic}
\usepackage{graphicx}
\usepackage[all,cmtip]{xy}

\newcommand{\CC}{\mathbb{C}}
\newcommand{\PP}{\mathbb{P}}
\newcommand{\ZZ}{\mathbb{Z}}

\newcommand{\eqn}{\begin{eqnarray}}
\newcommand{\feqn}{\end{eqnarray}}

\newtheorem{theorem}{Theorem}[section]
\newtheorem{lemma}[theorem]{Lemma}

\newcommand{\qedd}{\nobreak \ifvmode \relax \else
      \ifdim\lastskip<1.5em \hskip-\lastskip
      \hskip1.5em plus0em minus0.5em \fi \nobreak
      \vrule height0.75em width0.5em depth0.25em\fi}

\ifx\mb\undefined
\newcommand{\mb}[1]{{\mathbb #1}}
\fi

\date{\today}

\title{On the geometry of \boldmath{$\CC^3/\Delta_{27}$} and del Pezzo surfaces}

\author{Sergio L. Cacciatori\\Dipartimento di Fisica, Universit\`a degli Studi dell'Insubria,
Via Valleggio 11, 22100 Como, Italy\\E-mail address: sergio.cacciatori@uninsubria.it}

\author{Marco Compagnoni\\Dipartimento di Matematica, Politecnico di Milano,
Via Bonardi 9, 20133 Milano, Italy\\E-mail address: marco.compagnoni@polimi.it}

\abstract{We clarify some aspects of the geometry of a resolution of the orbifold $X=\CC^3/\Delta_{27}$, the noncompact complex manifold underlying the brane quiver standard model recently proposed by Verlinde and Wijnholt. We explicitly realize a map between $X$ and the total space of the canonical bundle over a degree 1 quasi del Pezzo surface, thus defining a desingularization of $X$. Our analysis relys essentially on the relationship existing between the normalizer group of $\Delta_{27}$ and the Hessian group and on the study of the behaviour of the Hesse pencil of plane cubic curves under the quotient.
}

\begin{document}

\section{Introduction}\label{sec:intro}
String theory, being on the scenes from many years now, can no more be considered a novel theory. In spite of having overtaken its initial difficulties and undergone more then a revolution leading to amazing interconnections between physics and mathematics, it has not yet provided a natural top-down physical prediction that could make of it a testable theory. The successes in revealing duality relations underlying different physical models as a consequence of AdS/CFT correspondence or the great deal of mathematical conjectures inspired by mirror symmetry, surely ascribe to string theory all its scientific dignity. However, to be accepted as a true physical theory, string theory needs to be connected to predictive models which can be tested via experiment. For this reason in the last years D-brane engineering has been developed in order to realize four dimensional supersymmetric models reproducing almost realistic gauge theories in which gravity is decoupled (see for example the recent works on F-theory phenomenology, starting from \cite{Beasley:2008dc}, \cite{Beasley:2008kw}, \cite{Donagi:2008ca}, \cite{Donagi:2008kj}).

One of these models, an explicit realization of the supersymmetric Standard Model as a world-volume theory on a D3-brane, has been proposed some time ago in a paper by H. Verlinde and M. Wijnholt \cite{Verlinde:2005jr}. It is a simple quiver extension of the minimal supersymmetric Standard Model. The geometric dual is realized starting from the non Abelian orbifold $\CC^3/\Delta_{27}$, where $\Delta_{27}$ is the Heisenberg group of order $27$ obtained by the central extension of the abelian group $\ZZ_3\times \ZZ_3$ by $\ZZ_3$. Here we will not review all the construction of the model but only recall the relationship between the orbifold $\CC^3/\Delta_{27}$ and del Pezzo surfaces of degree one.

The authors of \cite{Verlinde:2005jr} observed that the orbifold $\CC^3/\Delta_{27}$ and the tautological cone over a del Pezzo surface $d\PP_8$ have the same quiver and therefore the same fractional brane configurations. Indeed the authors conjectured that the $dP_8$ singularities should actually arise as deformations of such non Abelian orbifold. This is the main argument they put forward in support of their hypothesis. The discrete group $\Delta_{27}$ is the non-abelian subgroup of $SL(3,\CC)$ generated by two elements:
\begin{eqnarray*}
g_1=\left(
\begin{array}{ccc}
1 & 0      & 0       \\
0 & \omega & 0       \\
0 & 0      & \omega^2
\end{array}\right),
\qquad
g_2=\left(
\begin{array}{ccc}
0 & 0 & 1\\
1 & 0 & 0\\
0 & 1 & 0
\end{array}\right),
\end{eqnarray*}
with $\omega=e^{\frac{2\pi i}{3}}$. We then see that the following combinations of coordinates
\begin{eqnarray}
x &=&  X Y Z  \cr
z &=&  X^3 + Y^3 + Z^3\label{combi}\\
y &=& (X^3 + \omega Y^3 + \omega^2 Z^3)(X^3 + \omega^2 Y^3 + \omega Z^3)\cr
w &=& (X^3 + \omega Y^3 + \omega^2 Z^3)^3 \nonumber
\end{eqnarray}
are invariant under the action of $\Delta_{27}$. Actually they define a map from $\CC^3/\Delta_{27}$ to the tautological bundle over the weighted projective space $\PP_{1,1,2,3}$. The coordinates $(x,z,y,w)$ defined in (\ref{combi}) satisfy the homogeneous equation
\begin{eqnarray*}
w^2 + y^3 - 27 w x^3  + w z^3  -3 w yz = 0.
\end{eqnarray*}
Thus, they conclude that the orbifold $\CC^3/\Delta_{27}$ is isomorphic to the cone over such a singular hypersurface in $\PP_{1,1,2,3}$. The well known fact that all smooth hypersurfaces of degree six in $\PP_{1,1,2,3}$ are isomorphic to a $dP_8$ surface is a strong evidence for the conjecture.

In this article we would like to clarify the details of this problem giving the exact correspondence between the two geometries. An initial discussion has appeared in \cite{Artebani-Dolgachev} and we will start our analysis taking their result into account. We expect that the correspondence outlined here will be recognized as very helpful for further developments of the Verlinde-Wijnholt model.

As a warm up that elucidates the phylosophy behind section \ref{sec:toric}, we will focus on a very similar toric model which is $\CC^3/\ZZ_3\times \ZZ_3$. This is a cone over a singular cubic surface in $\PP^3$ possesing several resolutions, all related by flops, which can be easily determined via toric methods. A special one among them can be identified as the total space of the canonical bundle over a quasi del Pezzo surface of degree 3, that is a surface obtained by blowing up three points on $\PP^2$ twice. The interpretation of this resolution in terms of the given surface suggests a strategy which indeed works also for $X=\CC^3/\Delta_{27}$ where toric methods are no more available.

Next in section \ref{sec:delta} we analyze the non toric geometry. According to the correspondence conjectured by Verlinde and Wijnholt one expects to relate the resolution to a quasi del Pezzo surface $d\PP_8$. The main idea is as follows: \\
The initial singular variety is the cone over $\PP^2/\tilde \Gamma$, where $\tilde \Gamma$ is the quotient of $\Gamma=\Delta_{27}$ by its center, which is $\ZZ_3$. It has four degenerate orbits on $\PP^2$ with non trivial stabilizer yielding the four singular point on $\PP^2/\tilde \Gamma$. Now, a $d\PP_8$ surface is obtained by blowing up eight points in general positions (see later) on $\PP^2$. But all the cubics on $\PP^2$ passing through the eight points also have in common a ninth point. When these nine points are in a particular position, that is when they stay in a square array, then the corresponding set of cubics is called a Hesse pencil. Such an Hesse pencil contains exactly four degenerate cubics, each one being the union of three lines. Let us consider the quasi del Pezzo surface $S=Bl_{p_1,\ldots,p8} (\PP^2)$ obtained by blowing up eight of the nine base points. Unlike a full del Pezzo surface, in this case the anticanonical divisor $-K_S$ is not ample and the anticanonical maps instead of giving a projective embedding for $S$ in $\PP_{1,1,2,3}$, they shrunk to a point the components of the four singular cubics which do not pass through the ninth base point of the pencil. As a consequence the resulting singular variety can be identified with $\PP^2/\tilde \Gamma$. The anticanonical maps actually give a desingularizing map from the quasi del Pezzo surface $S$ (also called singular del Pezzo surface in \cite{Verlinde:2005jr}) to $\PP^2/\tilde \Gamma$. One arrives at this conclusion by realizing that pinching out eight base points, the Hesse pencil will substantially describe both $\PP^2$ and $\PP^2/\tilde \Gamma$ as elliptic fibrations over $\PP^1$. This proof makes use of the identification of the Hessian group with the image of the normalizer of $\Gamma$ in $PGL(3,\CC)$.

\

{\sf{This article is based on notes provided to us by professor Bert van Geemen, to whom we are indebted.}}

\section{The toric toy model}\label{sec:toric}
The simplest orbifold that admits as crepant resolution the total space of the canonical divisor over a del Pezzo surface is $\CC^3/\ZZ_3$ case. The mathematics and physics of this orbifold has been intensively studied (see for example \cite{Bouchard:2007nr} and reference herein). It is a toric variety isomorphic to the complex cone over the projective plane $\PP^2$ and its only crepant resolution is the total space of the canonical bundle over $\PP^2$, the del Pezzo surface of degree nine.

In this section we briefly study the geometry of the three dimensional orbifold $M=\CC^3/(\ZZ_3\times\ZZ_3)$. The analysis resembles that of the  next section for $\CC^3/\Delta_{27}$; on the other hand $M$ is a toric variety therefore we can use the powerful tools of toric geometry. Here we sketch the main ideas and results, referring to \cite{FultonT,Oda} for all the details of the computation.
\subsection{The orbifold $\CC^3/(\ZZ_3\times\ZZ_3)$}
The group $\ZZ_3\times\ZZ_3$ has two generators:
\begin{equation}
g_1=\left(
\begin{array}{ccc}
1 & 0      & 0       \\
0 & \omega & 0       \\
0 & 0      & \omega^2\\
\end{array}\right),
\qquad
g_2=\left(
\begin{array}{ccc}
\omega & 0      & 0     \\
0      & \omega & 0     \\
0      & 0      & \omega\\
\end{array}\right),
\end{equation}
with $\omega=e^{\frac{2\pi i}{3}}$.
It is Abelian, therefore $M$ is a toric variety. Its fan $\Delta$ is generated by the vectors
\begin{eqnarray}
v_1=\left(\begin{array}{c} -1\\1\\1 \end{array} \right) \ , \qquad
v_2=\left(\begin{array}{c} -1\\-2\\1 \end{array} \right) \ , \qquad
v_3=\left(\begin{array}{c} 2\\1\\1 \end{array} \right) \ ,
\end{eqnarray}
in $N\simeq\ZZ^3$ and in figure \ref{Fan-orb} we give the two dimensional intersection of $\Delta$ with the plane $z=1$.
\begin{figure}[hbtp]
\centering
\resizebox{5.8cm}{!}{
\includegraphics[viewport=200 500 420 720,width=5cm,clip]{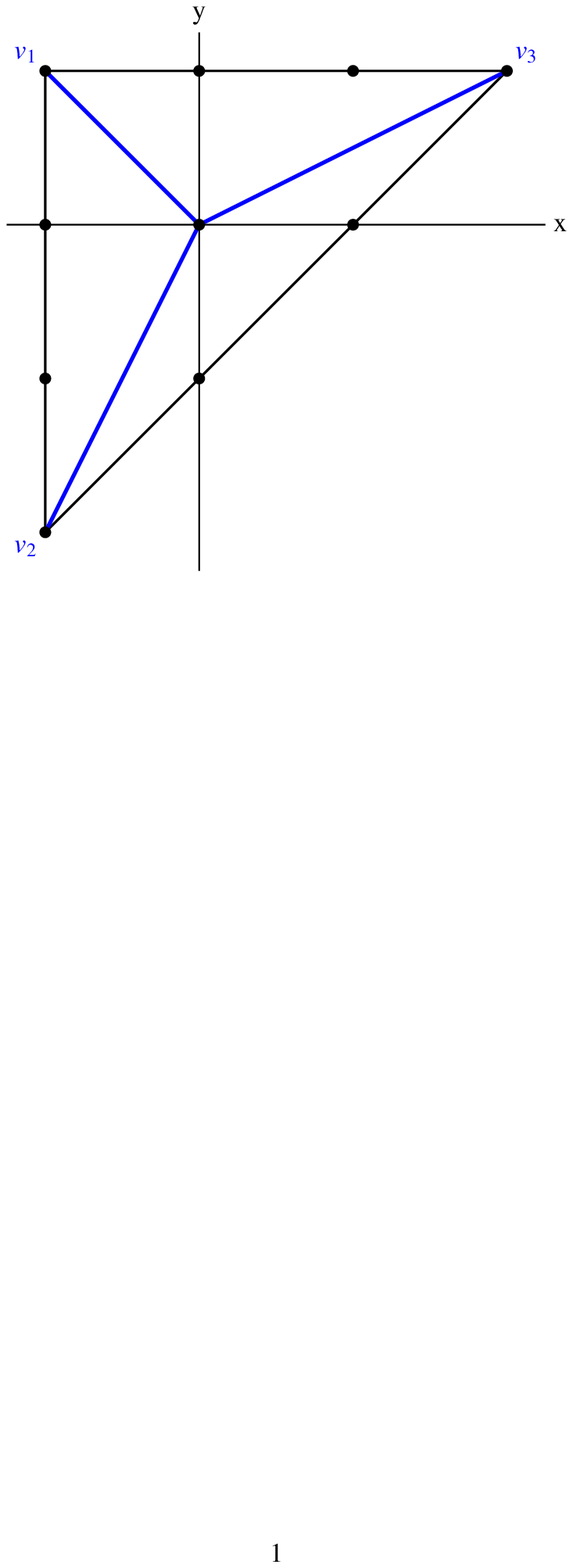}}
\caption{The fan of the orbifold $\CC^3/\ZZ_3\times \ZZ_3$. As a two dimensional fan it is the one of $\PP^2/\ZZ_3$.}
\label{Fan-orb}
\end{figure}
There are four monomials of degree three that are invariant under the action of $\ZZ_3\times\ZZ_3$:
\begin{eqnarray}
x & =  X Y Z,  \cr
y & =  X^3,    \label{orb2cubic}\\
z & =  Y^3,    \cr
w & =  Z^3.    \nonumber
\end{eqnarray}
They define a map from the orbifold $\CC^3/(\ZZ_3\times\ZZ_3)$ to the tautological bundle over the projective space $\PP^3$. Moreover the monomials satisfy the homogeneous relation
\begin{eqnarray}
x^3- y z w = 0.\label{singcubic}
\end{eqnarray}
Actually they define an isomorphism between the orbifold and the cone over the singular cubic surface defined by equation \ref{singcubic}. It is well known that any smooth cubic surface in $\PP^3$ is a $dP_6$, a del Pezzo surface of degree three isomorphic to the blown up of $\PP^2$ in six points in general position (see for example chapter 5 of \cite{Hartshorne}). Therefore the singular surface \ref{singcubic} should be related to some sort of limit in the moduli space of $dP_6$ surfaces.

$M$ is a non compact Calabi-Yau threefold with non isolated singularities. We have several crepant resolutions of such an orbifold, related by flop transitions (see \cite{noi} for a short guide to toric desingularization). We are interested in the one associated to the fan in figure \ref{Fan-orbres}.
\begin{figure}[hbtp]
\centering
\resizebox{5.8cm}{!}{
\includegraphics[viewport=200 500 420 730,width=6cm,clip]{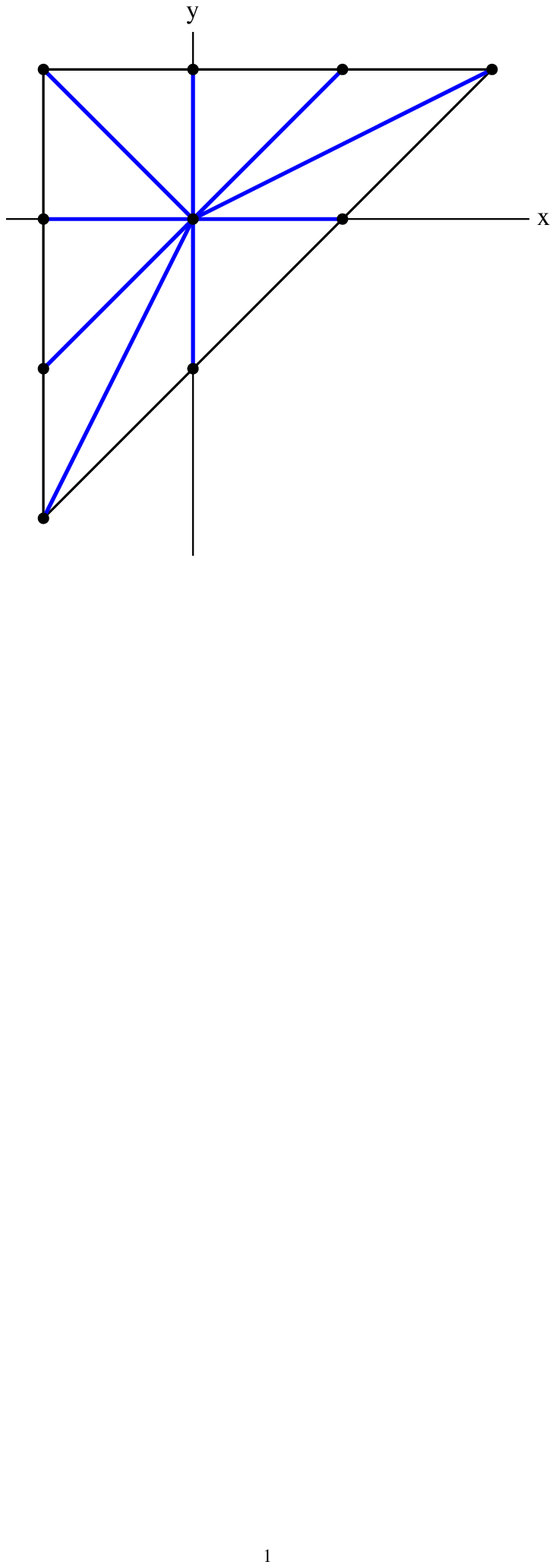}}
\caption{The fan of a crepant resolution of the orbifold $\CC^3/\ZZ_3\times \ZZ_3$. As a two dimensional fan it gives the minimal
resolution of $\PP^2/\ZZ_3$.} \label{Fan-orbres}
\end{figure}\noindent
In the next section we will study the relation between such a resolution and del Pezzo surfaces.

\subsection{The quotient $\PP^2/\ZZ_3$ and its minimal resolution}
The orbifold $M=\CC^3/(\ZZ_3\times\ZZ_3)$ is the tautological cone over the toric quotient surface $\PP^2_3=\mathbb{P}^2/\mathbb{Z}_3$, i.e. the total space of the canonical fiber bundle $K_{\PP^2_3}$ with the zero section shrunk to a point. The fan of $\PP^2_3$ is the same as in figure \ref{Fan-orb} if we think to it as a two dimensional fan.

The action of the generator $g$ of $\ZZ_3$ on the homogeneous coordinates of $\PP^2$ is:
\begin{equation}
g\cdot (z_1:z_2:z_3)=(z_1:\omega z_2:\omega^2 z_3),
\qquad \omega^3=1.
\end{equation}
There are three fixed points on $\PP^2$ at
\begin{equation}
p_1\equiv (1:0:0), \qquad p_2\equiv (0:1:0), \qquad p_3\equiv (0:0:1),
\end{equation}
the toric invariant points. Using local charts near the $p_i$, it is easy to see that the action of $g$ is given by diag$(\omega,\omega^2)$. Therefore the quotient surface has three singular points of type $A_2$. Since the orbifold $M$ is locally the product $\PP^2_3\times A^1_\CC$, in correspondence of the images of the $p_i$ we have three curves of $A_2$ singular points in $M$.
\begin{figure}[hbtp]
\centering
\resizebox{13cm}{!}{
\includegraphics[viewport=120 500 560 720,width=10cm,clip]{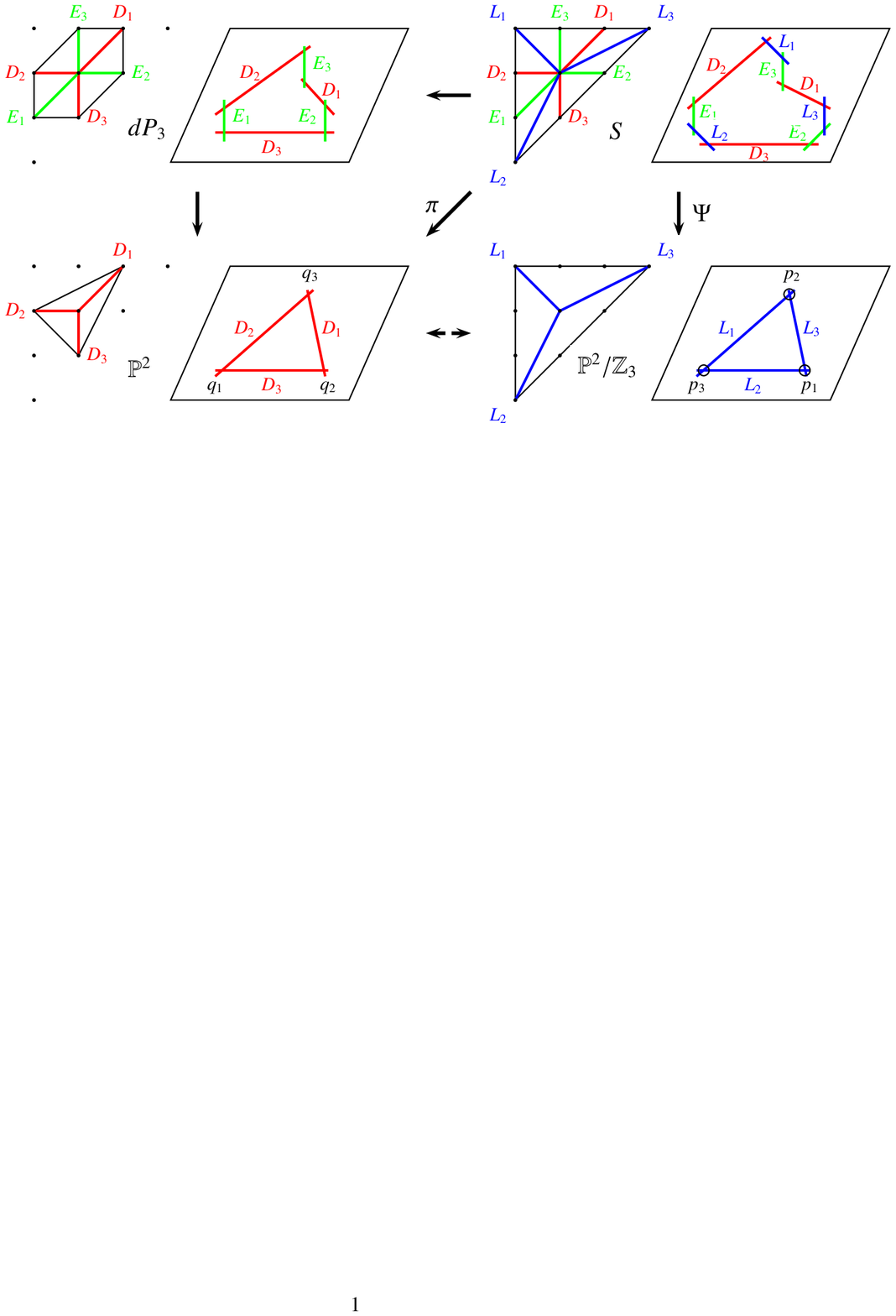}}
\caption{The minimal resolution of $\PP^2/\ZZ_3$.}
\label{Resolution}
\end{figure}\noindent
The resolution of the singularities in $\PP^2_3$ is easily performed by means of the standard toric methods through a double blow-up of the singular points. We obtain the smooth toric surface $S$ whose fan is depicted in figure \ref{Fan-orbres}. Taking the total space of the canonical bundle $K_S$ we obtain a crepant resolution of the initial orbifold $M$. The birationality of the surface $S$ with $\PP^2$ can be verified by means the blow up map $\pi$ that we depict in figure \ref{Resolution}. For simplicity we use the same name for the strict transforms of the relevant curves in the different varieties. Starting from $\PP^2$ we blow up the three toric invariant points $q_1=(1:0:0),\ q_2=(0:1:0),\ q_3=(0:0:1)$, obtaining the toric del Pezzo surface $dP_3$. The three exceptional divisors $E_i$ on $dP_3$ separate the three toric invariant lines $D_i$ on $\PP^2$ defined by $x_i=0$. Blowing up three intersections of the $E_i$ with the strict transforms of the $D_i$ in $dP_3$ we obtain the surface $S$. This is a degenerate limit in the moduli space of del Pezzo surfaces of degree three, it is obtained by blowing up $\PP^2$ in three couples of infinitely nearby points.

In $S$ we have the following intersection table:
\TABLE[h]{
\begin{tabular}{c|ccccccccc}
\phantom{x} & $L_1$ & $L_2$ & $L_3$ & $E_1$ & $D_2$ & $E_2$ & $D_3$ & $E_3$ & $D_1$ \\
\hline
  $L_1$     &   -1  &   0   &   0   &   0   &   1   &   0   &   0   &   1   &   0   \\
  $L_2$     &   0   &   -1  &   0   &   1   &   0   &   0   &   1   &   0   &   0   \\
  $L_3$     &   0   &   0   &   -1  &   0   &   0   &   1   &   0   &   0   &   1   \\
  $E_1$     &   0   &   1   &   0   &   -2  &   1   &   0   &   0   &   0   &   0   \\
  $D_2$     &   1   &   0   &   0   &   1   &   -2  &   0   &   0   &   0   &   0   \\
  $E_2$     &   0   &   0   &   1   &   0   &   0   &   -2  &   1   &   0   &   0   \\
  $D_3$     &   0   &   1   &   0   &   0   &   0   &   1   &   -2  &   0   &   0   \\
  $E_3$     &   1   &   0   &   0   &   0   &   0   &   0   &   0   &   -2  &   1   \\
  $D_1$     &   0   &   0   &   1   &   0   &   0   &   0   &   0   &   1   &   -2  \\
\end{tabular}
\caption{Intersection table of $S$}
\label{IntPairingtoy}
}\\
The intersection graph for each pair of curves $(E_i,D_{i+1})$ is $A_2$. The resolution map from $S$ to $\PP^2_3$ sends such curves to the three singular points of type $A_2$.

The map from $S$ to the singular cubic surface \ref{singcubic} is given by the anticanonical sections. The anticanonical divisor of a $dP_6$ is very ample and its sections give an embedding of the surface in $\PP^3$. The canonical divisor of the surface $S$ is $K_S=-\sum D_i+E_i+L_i$, which has zero intersection with the curves $(E_i,D_i)$. Hence in this case $-K_S$ does not satisfy the second requirement of the Nakai-Moishezon criterion (see section \ref{generalities}) and it is not ample. The anticanonical global sections are constant on the above curves, so they get blown down by the anticanonical map.

We can obtain the explicit anticanonical map using toric methods (see \cite{Oda}). The polytope in the dual lattice $M$ related to the anticanonical divisor $-K_S$ is depicted in figure \ref{AnticanonicalPolytope}.
\begin{figure}[hbtp]
\centering
\resizebox{5cm}{!}{
\includegraphics[viewport=260 620 420 790,width=5cm,clip]{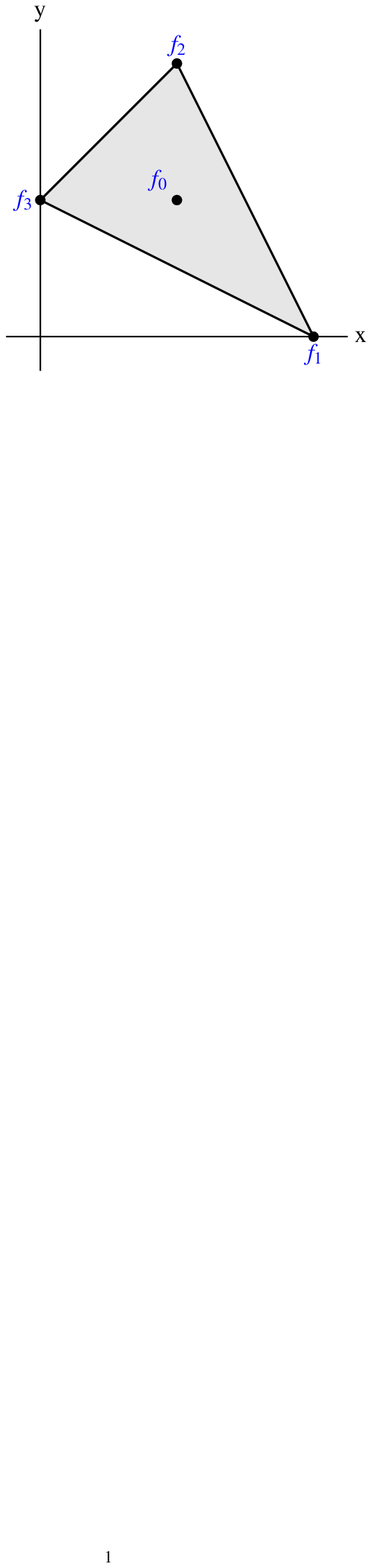}}
\caption{The polytope of $-K_S$}
\label{AnticanonicalPolytope}
\end{figure}\noindent
A basis for $H^0(S,-K_S)$ is given by the four monomials in the polytope:
\begin{align}
f_0 & =  X Y,   \cr
f_1 & =  X^2,   \label{S2cubic}\\
f_2 & =  X Y^2, \cr
f_3 & =  Y.     \nonumber
\end{align}
They define the equivariant map:
\begin{eqnarray}\label{psiregular33}
\Psi: \quad & S & \rightarrow \PP^3 \\
            & q & \mapsto (f_0(q):f_1(q):f_2(q):f_3(q)) \nonumber
\end{eqnarray}
The image of $\Psi$ is the singular cubic surface \ref{singcubic}:
\begin{eqnarray}
x^3- y z w = 0.\label{singcubic2}
\end{eqnarray}
This is a strictly semistable cubic surface with three double rational points of type $A_2$, which arises naturally in the context of the complex ball uniformization of the moduli space of cubic surfaces, when boundary points are included (see \cite{Dolgachev:1}, section 9 and \cite{Allcock}).

\subsection{Conclusion}
The singular orbifold $M=\CC^3/(\ZZ_3\times\ZZ_3)$ is the tautological cone over the singular quotient surface $\PP^2_3=\PP^2/\ZZ_3$. The threefold $M$ has three curves of singular points of $A_2$ type with common intersection at the origin. The quotient $\PP^2_3$ is isomorphic to a singular cubic surface in $\PP^3$ and it admits the quasi del Pezzo surface $S$ as its minimal resolution with anticanonical desingularization map $\Psi$. We have the following commutative diagram:\\[-1mm]

\hspace{5.4cm}
\xymatrix{
\PP^2 \ar@{->}[d]
&S \ar[l]_{\ \pi} \ar@{->}[ld]^\Psi\\
\PP^2/\tilde\Gamma}\\[3mm]
The total space of the canonical bundle $K_S$ and the map $\Psi$ give a crepant resolution of the starting orbifold $M$.

\section{The geometry of $\CC^3/\Gamma$}\label{sec:delta}
In this section we study the relation between the non Abelian orbifold $\CC^3/\Gamma$ and del Pezzo surfaces, obtaining a very similar structure to the one of the preceding section.

\subsection{The Heisenberg group $\Gamma$ and its action on $\CC^3$}
\label{c3/d27}
Here we will look at the geometry of the singular quotient $\CC^3/\Gamma$.
\subsubsection{The Heisenberg group $\Gamma$}
The matrices group $\Gamma:=\Delta_{27}\subset SL(3,\CC)$ has two generators:
\begin{equation}
g_1=\left(
\begin{array}{ccc}
1 & 0      & 0       \\
0 & \omega & 0       \\
0 & 0      & \omega^2\\
\end{array}\right),
\qquad
g_2=\left(
\begin{array}{ccc}
0 & 0 & 1\\
1 & 0 & 0\\
0 & 1 & 0\\
\end{array}\right),
\end{equation}
with $\omega=e^{\frac{2\pi i}{3}}$. They satisfy the relations:
\begin{equation}
g_1^3=g_2^3=I,\qquad g_2 g_1 = \omega^2 g_1 g_2.
\end{equation}
The center of $\Gamma$, its maximal abelian subgroup, is
\begin{equation}
C=\{I,\omega I,\omega^2 I\}.
\end{equation}
The abelianization of $\Gamma$ is $\tilde\Gamma:=\Gamma/C\simeq\ZZ_3\times\ZZ_3$. Thus $\Gamma$ is the Heisenberg group of order $27$, i.e. the non abelian central extension of the group $\ZZ_3\times\ZZ_3$ by $\ZZ_3$.

\subsubsection{The quotient $\PP^2/\tilde\Gamma$ and its relation with $\CC^3/\Gamma$}
The group $\tilde\Gamma$ is also the image of $\Gamma$ in $PGL(3,\CC)$, the group of automorphisms of the projective plane. It has a natural action on $\PP^2$ which is strictly related to the one of $\Gamma$ on $\CC^3$. $\tilde\Gamma$ has four cyclic subgroup isomorphic to $\ZZ_3$:
\begin{align}
<g_1>&=\{g_1,g_1^2,I\}
&
g_1&=\left(
\begin{array}{ccc}
1 & 0      & 0       \\
0 & \omega & 0       \\
0 & 0      & \omega^2\\
\end{array}\right),\cr
<g_2>&=\{g_2,g_2^2,I\}
&
g_2&=\left(
\begin{array}{ccc}
0 & 0 & 1\\
1 & 0 & 0\\
0 & 1 & 0\\
\end{array}\right),\\
<g_1g_2>&=\{g_1g_2,g_1^2g_2^2,I\}
&
g_1g_2&=\left(
\begin{array}{ccc}
0      & 0        & 1 \\
\omega & 0        & 0 \\
0      & \omega^2 & 0 \\
\end{array}\right),\cr
<g_1^2g_2>&=\{g_1^2g_2,g_1g_2^2,I\}
&
g_1^2g_2&=\left(
\begin{array}{ccc}
0        & 0      & 1 \\
\omega^2 & 0      & 0 \\
0        & \omega & 0 \\
\end{array}\right).\nonumber
\end{align}
The elements of each subgroup have common eigenspaces in $\CC^3$:
\begin{align}
g_1\cdot \left(
\begin{array}{c}
1 \\ 0 \\ 0 \\
\end{array}\right)
&=\left(
\begin{array}{c}
1 \\ 0 \\ 0 \\
\end{array}\right),
&
g_1\cdot \left(
\begin{array}{c}
0 \\ 1 \\ 0 \\
\end{array}\right)
&=\omega\left(
\begin{array}{c}
0 \\ 1 \\ 0 \\
\end{array}\right),
&
g_1\cdot\left(
\begin{array}{c}
0 \\ 0 \\ 1 \\
\end{array}\right)
&=\omega^2\left(
\begin{array}{c}
0 \\ 0 \\ 1 \\
\end{array}\right),\cr
g_2\cdot\left(
\begin{array}{c}
1 \\ 1 \\ 1 \\
\end{array}\right)
&=\left(
\begin{array}{c}
1 \\ 1 \\ 1 \\
\end{array}\right),
&
g_2\cdot\left(
\begin{array}{c}
1 \\ \omega^2 \\ \omega \\
\end{array}\right)
&=\omega\left(
\begin{array}{c}
1 \\ \omega^2 \\ \omega \\
\end{array}\right),
&
g_2\cdot\left(
\begin{array}{c}
1 \\ \omega \\ \omega^2 \\
\end{array}\right)
&=\omega^2\left(
\begin{array}{c}
1 \\ \omega \\ \omega^2 \\
\end{array}\right),\\
g_1g_2\cdot\left(
\begin{array}{c}
1 \\ \omega \\ 1 \\
\end{array}\right)
&=\left(
\begin{array}{c}
1 \\ \omega \\ 1 \\
\end{array}\right),
&
g_1g_2\cdot\left(
\begin{array}{c}
1 \\ 1 \\ \omega \\
\end{array}\right)
&=\omega\left(
\begin{array}{c}
1 \\ 1 \\ \omega \\
\end{array}\right),
&
g_1g_2\cdot\left(
\begin{array}{c}
\omega \\ 1 \\ 1 \\
\end{array}\right)
&=\omega^2\left(
\begin{array}{c}
\omega \\ 1 \\ 1 \\
\end{array}\right),\cr
g_1^2g_2\cdot\left(
\begin{array}{c}
1 \\ \omega^2 \\ 1 \\
\end{array}\right)
&=\left(
\begin{array}{c}
1 \\ \omega^2 \\ 1 \\
\end{array}\right),
&
g_1^2g_2\cdot\left(
\begin{array}{c}
\omega^2 \\ 1 \\ 1 \\
\end{array}\right)
&=\omega\left(
\begin{array}{c}
\omega^2 \\ 1 \\ 1 \\
\end{array}\right),
&
g_1^2g_2\cdot\left(
\begin{array}{c}
1 \\ 1 \\ \omega^2 \\
\end{array}\right)
&=\omega^2\left(
\begin{array}{c}
1 \\ 1 \\ \omega^2 \\
\end{array}\right).\nonumber
\end{align}
To any eigenspace it corresponds a fixed point in $\PP^2$ of the relative cyclic subgroup of $\tilde\Gamma$:
\begin{equation}
\label{degorb}
\begin{split}
\text{Fix}(<g_1>):      & \quad(1:0:0),        \quad (0:1:0),            \quad (0:0:1),\\
\text{Fix}(<g_2>):      & \quad(1:1:1),        \quad(1:\omega^2:\omega), \quad(1:\omega:\omega^2),\\
\text{Fix}(<g_1g_2>):   & \quad(1:\omega:1),   \quad(1:1:\omega),        \quad(\omega:1:1),\\
\text{Fix}(<g_1^2g_2>): & \quad(1:\omega^2:1), \quad(\omega^2:1:1),      \quad(1:1:\omega^2).
\end{split}
\end{equation}
Any set of fixed points \ref{degorb} of a given cyclic subgroup supports a degenerate orbit of order three of $\tilde\Gamma$ in $\PP^2$. These are the only orbits in the projective plane with non-trivial stabilizer in $\tilde\Gamma$. Locally at each fixed point the action of the stabilizer subgroup is given by $\text{diag}(\omega,\omega^2)$. This is very easy to see for the fixed points of $g_1$. For example near $(0:0:1)$ we can use the local chart:
\begin{eqnarray}
\PP^2   &&\longrightarrow \CC^2\\
(X:Y:Z) &&\longmapsto     \left(u=\frac{X}{Z},v=\frac{Y}{Z}\right).\nonumber
\end{eqnarray}
Thus the local action of $g_1$ is
\begin{equation}
g_1\cdot(u,v)=(\omega\, u,\omega^2\, v).
\end{equation}
It follows that the singular locus of the quotient surface $\PP^2/\tilde\Gamma$ consists of four singular points of type $A_2$, i.e. near the singular points the surface is isomorphic to the orbifold $\CC^2/\ZZ_3$.

The original orbifold $\CC^3/\Gamma$ is the cone over the singular surface $\PP^2/\tilde\Gamma$ and in particular the four singular points correspond to four lines passing through the origin, which in parametric form are:
\begin{equation}
\begin{split}
&\mathcal{C}_1: \qquad (t,0,0)\sim(0,t,0)\sim(0,0,t),\\
&\mathcal{C}_2: \qquad
(t,t,t)\sim(t,\omega^2\,t,\omega\,t)\sim(t,\omega\,t,\omega^2\,t),\\
&\mathcal{C}_3:\qquad
(t,\omega\,t,t)\sim(t,t,\omega\,t)\sim(\omega\,t,t,t),\\
&\mathcal{C}_4:\qquad
(t,\omega^2\,t,t)\sim(\omega^2\,t,t,t)\sim(t,t,\omega^2\,t).
\end{split}
\end{equation}

\subsubsection{The Verlinde-Wijnholt isomorphism}
In this section we prove that the map \ref{combi} is an isomorphism of algebraic varieties.

\begin{theorem}
The ring homomorphism associated to \ref{combi}
\begin{eqnarray}
\phi^*:\quad \CC[x,z,y,w] && \longrightarrow \quad \CC[X,Y,Z]^\Gamma\\
             (x,z,y,w)    && \longrightarrow \quad (f_0,f_1,f_2,f_3)\nonumber
\end{eqnarray}
with
\begin{eqnarray}
&& f_0  =  X Y Z,   \cr
&& f_1  =  X^3 + Y^3 + Z^3,   \\
&& f_2  =  (X^3 +\omega Y^3 +\omega^2 Z^3)(X^3 +\omega^2 Y^3 + \omega Z^3), \cr
&& f_3  =  (X^3 + \omega Y^3 + \omega^2 Z^3)^3,     \nonumber
\end{eqnarray}
defines a ring isomorphism:
\begin{equation}
\CC[X,Y,Z]^\Gamma\simeq\frac{\CC[x,z,y,w]}{(w^2+y^3-27wx^3-3wyz+wz^3)}.
\end{equation}
\end{theorem}

The proof of this theorem uses the following lemmas.

\begin{lemma}
\label{Lemma1}
The morphism
\begin{eqnarray}
\phi_R: \quad\qquad \CC^3 && \longrightarrow \quad\CC^3\\
                  (X,Y,Z) && \longmapsto\quad(f_0,f_1,f_2)\nonumber
\end{eqnarray}
is surjective.
\end{lemma}

\begin{proof}
Any point with coordinates $(x,z,y)$ having $x\neq 0$ has preimages with coordinates $(X,Y,Z)$ the solutions of the following algebraic equations:
\begin{eqnarray}
&&Y^9-z\;Y^6+\frac{z^2-y}{3}\;Y^3-x^3=0,\cr
&&Z^6+(Y^3-z)\;Z^3+\left(Y^6-zY^3+\frac{z^2-y}{3}\right)=0,\\
&&X-\frac{x}{YZ}=0.\nonumber
\end{eqnarray}
The case $x=0$ is similar.
\end{proof}

\begin{lemma}
\label{Lemma2}
The graded ring
\begin{equation}
\frac{\CC[x,z,y,w]}{(w^2+y^3-27wx^3-3wyz+wz^3)}
\end{equation}
with gradation
\begin{equation}
deg(x^a z^b y^c w^d):=3a+3b+6c+9d
\end{equation}
is the direct sum of vector spaces of dimension
\begin{equation}
\text{dim}\left(\frac{\CC[x,z,y,w]}{(w^2+y^3-27wx^3-3wyz+wz^3)}\right)_n=
\left\{
\begin{array}{cc}
0 & \text{if}\ n \neq 0\ \text{mod}\ 3,\\
1+\frac{d(d+1)}{2} & \text{if}\ n=3d.
\end{array}
\right.
\end{equation}
\end{lemma}

\begin{proof}
We first observe that
\begin{equation}
\label{sommadiretta}
\frac{\CC[x,z,y,w]}{(w^2+y^3-27wx^3-3wyz+wz^3)}\simeq
\CC[x,z,y]\oplus\CC[x,z,y]w,
\end{equation}
thus as vector space it admits the monomial basis
\begin{equation}
x^a z^b y^c w^d \qquad (a,b,c)\in \ZZ_{\geq 0}^3 \quad d\in\{0,1\}.
\end{equation}
{F}rom \ref{sommadiretta} it follows that
\begin{equation}
\text{dim}\left(\frac{\CC[x,z,y,w]}{(w^2+y^3-27wx^3-3wyz+wz^3)}\right)_n=
\text{dim}\left(\CC[x,z,y]\right)_n+\text{dim}\left(\CC[x,z,y]\right)_{n-9}.
\end{equation}
We have to calculate $\text{dim}\left(\CC[x,z,y]\right)_n$, that is the number of monomials $x^a z^b y^c$ with $3a+3b+6c=n$. Note that
\begin{eqnarray}
&& \frac{1}{1-x^3t^3}\frac{1}{1-z^3t^3}\frac{1}{1-y^6t^6}\cr
&& \qquad\qquad\ = (1+x^3t^3+\ldots+x^{3a}t^{3a}+\ldots)(1+\ldots+z^{3b}t^{3b}+\ldots)
    (1+\ldots+y^{6c}t^{6c}+\ldots)\cr
&& \qquad\qquad\ = \sum_{a,b,c\geq 0}x^{3a}z^{3b}y^{6c}t^{3a+3b+6c}
 = \sum_{n=0}^{\infty}\left(\sum_{3a+3b+6c=n}x^{3a}z^{3b}y^{6c}\right)
    t^n.
\end{eqnarray}
Putting $x=z=y=1$ we get the generating function for $\text{dim}\left(\CC[x,z,y]\right)_n$:
\begin{equation}
\sum_{n=0}^{\infty}\text{dim}\left(\CC[x,z,y]\right)_n t^n =
\frac{1}{1-t^3}\frac{1}{1-t^3}\frac{1}{1-t^6} =
\frac{1}{(1-t^3)^3}\frac{1}{1+t^3}.
\end{equation}
Thus
\begin{eqnarray}
\sum_{n=0}^{\infty}
\text{dim}\left(\frac{\CC[x,z,y,w]}{(w^2+y^3-27wx^3-3wyz+wz^3)}\right)_n t^n
& =& \sum_{n=0}^{\infty}\text{dim}\left(\CC[x,z,y]\right)_n t^n +
                       \text{dim}\left(\CC[x,z,y]\right)_{n-9} t^n\cr
& =& \frac{1}{(1-t^3)^3}\frac{1}{1+t^3}+\frac{t^9}{(1-t^3)^3}\frac{1}{1+t^3}\cr
& =& \frac{t^6-t^3+1}{(1-t^3)^3}\cr
& =& (t^6-t^3+1)\left(\sum_{n=0}^{\infty}\left(^{n+2}_{2}\right)t^{3n}\right)\cr
& =& \sum_{n=0}^{\infty}\left(\left(^{n+2}_{2}\right)-\left(^{n+1}_{2}\right)
    +\left(^{n}_{2}\right)\right)t^{3n}\cr
& =& \sum_{n=0}^{\infty}\left(1+\frac{n(n+1)}{2}\right)t^{3n},
\end{eqnarray}
where the relation $\frac{1}{(1-t^3)^3}=\sum_{n=0}^{\infty}\left(^{n+2}_{2}\right) \,t^{3n}$ can be obtained by differentiating $\frac{1}{1-t}=\sum_{n=0}^{\infty}t^{n}$ twice. The statement follows.
\end{proof}

\begin{lemma}
\label{Lemma3}
The graded ring $\CC[X,Y,Z]^{\Gamma}$ of invariant polynomials under the action of $\Gamma$, with gradation
\begin{equation}
deg(X^A Y^B Z^C):=A+B+C\ ,
\end{equation}
is the direct sum of vector spaces of dimension
\begin{equation}
\text{dim}\left(\CC[X,Y,Z]^{\Gamma}\right)_n=
\left\{
\begin{array}{cc}
0 & \text{if}\ n \neq 0\ \text{mod}\ 3,\\
1+\frac{d(d+1)}{2} & \text{if}\ n=3d.
\end{array}
\right.
\end{equation}
\end{lemma}

\begin{proof}
Note that $\left(\CC[X,Y,Z]\right)_n$ is a finite dimensional representation $\rho_n$ of $\Gamma$, which admits a unique decomposition in irreducible representations:
\begin{equation}
\left(\CC[X,Y,Z]\right)_n=\rho_n=\rho_0^{a_{0,n}}\oplus
\rho_1^{a_{1,n}}\oplus\ldots.
\end{equation}
However the space of $\Gamma$-invariants polynomials is isomorphic to the trivial subrepresentation in $\rho_n$
\begin{equation}
\left(\CC[X,Y,Z]^{\Gamma}\right)_n=\bigoplus_{a_{0,n}}\rho_0
\end{equation}
and its dimension is the multiplicity $a_{0,n}$ of $\rho_0$.
The representation theory of finite groups says us that if we define the character function
\begin{align}
\chi_{\rho}: & \Gamma \longrightarrow \CC \\
             & g      \longmapsto     \text{tr} (\rho(g))\nonumber
\end{align}
the multiplicity $a_i$ of the irreducible representation $\rho_i$ in the decomposition of $\rho$ is
\begin{equation}
a_i=\frac{1}{|\Gamma|}\sum_{g\in\Gamma}\,\chi_\rho(g)\,\overline{\chi_{\rho_i}(g)}.
\end{equation}
For the trivial representation one has $\chi_{\rho_0}(g)=1$ for any $g\in\Gamma$, hence
\begin{equation}
a_{0,n}=\frac{1}{|\Gamma|}\sum_{g\in\Gamma}\,\chi_{\rho_n}(g)=\frac{1}{27}\sum_{g\in\Gamma}\text{tr}\,{\rho_n}(g).
\end{equation}
We need to determinate tr$\,\rho_n(g)$ for any $g\in\Gamma$. In case $g\notin C$ we can diagonalize it finding new monomials $(X',Y',Z')$ such that
\begin{equation}
g=\left(
\begin{array}{ccc}
1 & 0      & 0       \\
0 & \omega & 0       \\
0 & 0      & \omega^2\\
\end{array}\right).
\end{equation}
${\rho_n}'=\left(\CC[X',Y',Z']\right)_n$ is a representation equivalent to $\rho_n$, thus the characters are the same. So, it suffices to compute the trace of $\text{diag}(1,\omega,\omega^2)$ on $\rho_n$. Note that
\begin{eqnarray}
&&\frac{1}{1-\lambda_1 X}\frac{1}{1-\lambda_2 Y}\frac{1}{1-\lambda_3 Z}\cr
&&\qquad\qquad\ = (1+\lambda_1 X+\ldots+\lambda_1^A X^A+\ldots)
    (1+\ldots+\lambda_2^B Y^B+\ldots)(1+\ldots+\lambda_3^C Z^C+\ldots)\cr
&&\qquad\qquad\ = \sum_{A,B,C\geq 0}\lambda_1^A\lambda_2^B\lambda_3^C X^A Y^B Z^C.
\end{eqnarray}
The matrix $g=\text{diag}(\lambda_1,\lambda_2,\lambda_3)$ has trace $\sum_{A+B+C}\lambda_1^A\lambda_2^B\lambda_3^C$ on $\left(\CC[X,Y,Z]\right)_n$. Putting $X=Y=Z=t$ we thus get the generating function:
\begin{equation}
\sum_{n=0}^{\infty}\text{tr}\,\rho_n(g)\,t^n =
\sum_{n=0}^{\infty}\left(\sum_{A+B+C=n}\lambda_1^A\lambda_2^B\lambda_3^C
\right)t^n=
\frac{1}{1-\lambda_1 t}\frac{1}{1-\lambda_2 t}\frac{1}{1-\lambda_3 t}.
\end{equation}
In the case $g=\text{diag}(1,\omega,\omega^2)$ we obtain:
\begin{equation}
\sum_{n=0}^{\infty}\text{tr}\,\rho_n(g)\,t^n =
\frac{1}{1-t}\frac{1}{1-\omega t}\frac{1}{1-\omega^2 t}=
\frac{1}{1-t^3}=
\sum_{n=0}^{\infty}t^{3n} .
\end{equation}
The case $g=\omega^a I\in C$ is easier:
\begin{equation}
\text{tr}\,\rho_n(\omega^a I)=\omega^{an}\text{dim}\left(\CC[X,Y,Z]\right)_n
=\omega^{an}\left(^{n+2}_{2}\right),
\end{equation}
where the last equivalence can be obtained as the one for $\text{dim}\left(\CC[x,z,y]\right)_n$.\\
It follows that $\left(\CC[X,Y,Z]^{\Gamma}\right)_n=0$ if $n\neq 0\ \text{mod}\ 3$, as it should appear obvious since $\omega I$ acts as multiplication by $\omega^n$ on them. Instead for $n=3d$:
\begin{align}
\text{dim}\left(\CC[X,Y,Z]^\Gamma\right)_{3d}
&=\frac{1}{27}\left(1\cdot24+(1+\omega^{3d}+\omega^{6d})\left(^{3d+2}_{2}\right)\right)\cr
&=\frac{1}{27}\frac{27d^2+27d+54}{2}
=1+\frac{d(d+1)}{2}.
\end{align}
\end{proof}

We conclude the proof of the main theorem of the section.
\begin{proof}
As observed in \cite{Verlinde:2005jr}
\begin{equation}
(w^2+y^3-27wx^3-3wyz+wz^3)\subset\text{Ker}\,\phi^*.
\end{equation}
Thus we have the ring homomorphism (that we call again $\phi^*$)
\begin{equation}
\label{algmap}
\phi^*:\quad \frac{\CC[x,z,y,w]}{(w^2+y^3-27wx^3-3wyz+wz^3)}
           \longmapsto\CC[X,Y,Z]^\Gamma
\end{equation}

This map is injective. In fact for the lemma \ref{Lemma1} the map $\phi_R$ is surjective. This implies that the map
\begin{eqnarray}
\phi_R^*:\quad \CC[x,z,y] && \longrightarrow \quad \CC[X,Y,Z]^\Gamma\\
               (x,z,y)    && \longrightarrow \quad (f_0,f_1,f_2)\nonumber
\end{eqnarray}
between the rings is injective and there are not relations involving $x,z,y$. Moreover the image of $\CC^3/\Gamma$ under the geometric map $\phi$ is a hypersurface in $\CC^4$. It is contained in $(w^2+y^3-27wx^3-3wyz+wz^3)=0$, that is an irreducible hypersurface, thus the image of $\phi$ coincides with it. Hence $\phi^*$ is injective.

The map \ref{algmap} is factorizable in vector space morphisms
\begin{equation}
\phi^*_n:
\left(\quad\frac{\CC[x,z,y,w]}{(w^2+y^3-27wx^3-3wyz+wz^3)}\right)_{n}
\longmapsto\left(\CC[X,Y,Z]^\Gamma\right)_{n}.
\end{equation}
The map $\phi^*$ is surjective if $\phi^*_n$ is surjective for any $n$. $\text{Ker}(\phi^*)=\varnothing$ implies $\text{Ker}(\phi^*_n)=\varnothing$ for any $n$, hence
\begin{equation}
\text{dim Im}(\phi^*_n)=\text{dim}
\left(\quad\frac{\CC[x,z,y,w]}{(w^2+y^3-27wx^3-3wyz+wz^3)}\right)_{n}.
\end{equation}
Therefore $\phi^*$ is surjective if
\begin{equation}
\text{dim}
\left(\quad\frac{\CC[x,z,y,w]}{(w^2+y^3-27wx^3-3wyz+wz^3)}\right)_{n}
=\text{dim}\left(\CC[X,Y,Z]^\Gamma\right)_{n}
\end{equation}
for any $n$, that is the statement of the lemmas \ref{Lemma2}, \ref{Lemma3}.
\end{proof}

The isomorphism $\phi^*$ determines the geometric isomorphism $\phi$ that maps the orbifold $\CC^3/\Gamma$ in the hypersurface
\begin{equation}
w^2+y^3-27wx^3-3wyz+wz^3=0\ \subset\ \CC^4.
\end{equation}
In particular it sends the singular curves of the orbifold to the singular locus of the hypersuface, which by standard analysis is the union of four curves intersecting in the origin:
\begin{align}
\mathcal{C}_1 &\longmapsto (0,t,t^2,t^3),\cr
\mathcal{C}_2 &\longmapsto (t,3t,0,0),\\
\mathcal{C}_3 &\longmapsto (t,3\omega\,t,0,0),\cr
\mathcal{C}_4 &\longmapsto (t,3\omega^2\,t,0,0).\nonumber
\end{align}
Finally, we observe that the map $\phi$ defines an isomorphism of surfaces between the quotient $\PP^2/\tilde\Gamma$ and a singular hypersurface in a weighted projective space:
\begin{eqnarray}
\tilde\phi: & \PP^2/\tilde\Gamma  &\longrightarrow \PP_{1,1,2,3}\\
&(X:Y:Z)             &\longmapsto (x:z:y:w).\nonumber
\end{eqnarray}
In particular it sends the singular points to:
\begin{align}
&q_1=(1:0:0)\sim(0:1:0)\sim(0:0:1)  \longmapsto (0:1:1:1),\cr
&q_2=(1:1:1)\sim(1:\omega^2:\omega)\sim(1:\omega:\omega^2) \longmapsto (1:3:0:0),\\
&q_3=(1:\omega:1)\sim(1:1:\omega)\sim(\omega:1:1) \longmapsto (1:3\omega:0:0),\cr
&q_4=(1:\omega^2:1)\sim(\omega^2:1:1)\sim(1:1:\omega^2) \longmapsto (1:3\omega^2:0:0).\nonumber
\end{align}

\subsection{The automorphisms group of $\Gamma$ and the Hesse pencil of cubic curves}
\label{AutHP}
In this sections we study the automorphisms group of $\Gamma$ and we prove the following theorem:

\begin{theorem}
\label{theorHessian}
The largest subgroup of the automorphism group $PGL(3,\CC)$ of the projective plane $\PP^2$ that respects the quotient $\PP^2/\tilde\Gamma$ is the Hessian group.
\end{theorem}

This Hessian group naturally acts on the Hesse pencil of plane cubic curves. The analysis of the Hesse pencil behaviour under the quotient will suggest the right desingularization for the singular surface $\PP^2/\tilde\Gamma$ as a quasi del Pezzo surface.
\subsubsection{The automorphisms group of $\Gamma$}
We keep on the analysis of the orbifold by a deeper study of the properties of $\Gamma$. The normalizer $N\subset GL(3,\CC)$ of $\Gamma$ is the group defined by
\begin{equation}
N:=\{n\in GL(3,\CC):n\Gamma n^{-1}\subset\Gamma\}.
\end{equation}
It acts naturally on $\CC^3$ and it is the largest subgroup of $GL(3,\CC)$ which sends a $\Gamma$ orbit into another one. We define
\begin{equation}
D:=\{zI:z\in\CC^*\}
\end{equation}
the normal subgroup of $N$ of diagonal matrices in $GL(3,\CC)$. The image of $N$ in $PGL(3,\CC)$ is the quotient $\tilde N:=N/D$. It acts naturally on $\PP^2$ and it is the largest subgroup of $PGL(3,\CC)$ which sends a $\tilde\Gamma$ orbit into another one. $\tilde\Gamma$ is a normal subgroup of $\tilde N$ and the quotient group $\tilde N/\tilde\Gamma$ acts naturally on the quotient surface $\PP^2/\tilde\Gamma$. In fact, if $\bar x\in \PP^2/\tilde\Gamma$ has representative $x\in \PP^2$ and if $\bar n\in\tilde N / \tilde\Gamma$ has representative $n\in \tilde N$ then $\bar n\cdot\bar x=\overline{nx}$ is well defined because if $g,h\in\tilde\Gamma$ we have $n(gh)n^{-1}=g'\in \tilde\Gamma$, so $\overline{ng}\cdot\overline{hx}=\overline{g'nx}=\overline{nx}$.

\begin{theorem}
The group $\tilde N/\tilde\Gamma$ is isomorphic to $SL(2,\ZZ_3)$ and the group $\tilde N$ has order $216$.
\end{theorem}

\begin{proof}
Any element $n\in N$ defines an automorphism
\begin{align}
\phi_n:\qquad
& \Gamma  \stackrel{\simeq}{\longrightarrow} \Gamma\\
& g       \longmapsto                         ngn^{-1}.\nonumber
\end{align}
Such automorphisms group of Heisenberg group are studied in great generality in \cite{weil} (see also \cite{ligozat,yoshida,gerardin76,gerardin77,cms}). Any element in $\Gamma$ can be written uniquely as $\omega^k g_1^a g_2^b$, with $k,a,b\in\{0,1,2\}$. For any $n\in N$ the automorphism $\phi_n$ is determined by its action on the generators $g_1,g_2$ of $\Gamma$. Therefore any $\phi_n$ is determined by the elements $k,l,a,b,c,d\in \{0,1,2\}$:
\begin{equation}
\phi_n(g_1)=\omega^k g_1^a g_2^b,\qquad
\phi_n(g_2)=\omega^l g_1^c g_2^d.
\end{equation}
It is easy to verify that the map
\begin{eqnarray}
\chi:\qquad
& N  \longrightarrow GL(2,\ZZ_3)\\
& n  \longmapsto     \left(\begin{array}{cc} a & c\\ b & d \end{array}\right)\nonumber
\end{eqnarray}
is a homomorphism of groups.

Now we use the following lemmas:

\begin{lemma}
The homomorphism $\chi$ is a surjection from $N$ to $SL(2,\ZZ_3)$.
\end{lemma}

\begin{proof}
The image of $\chi$ is contained in $SL(2,\ZZ_3)$. We use the property
\begin{equation}
g_2 g_1 = \omega^2 g_1 g_2 \qquad \Rightarrow \qquad
g_2^n g_1^m = \omega^{2nm} g_1^m g_2^n
\end{equation}
and the fact that for any element $c\in C$ in the center of $\Gamma$ we have $\phi_n(c)=c$. We get
\begin{equation}
\phi_n(g_2)\phi_n(g_1)=\omega^2\phi_n(g_1)\phi_n(g_2)
\end{equation}
and as
\begin{equation}
\left\{
\begin{array}{ccccccccc}
\phi_n(g_2)\phi_n(g_1) &\ & = &\ & \omega^{k+l} g_1^c g_2^d g_1^a g_2^b
&\ & = &\ & \omega^{k+l+2ad} g_1^{a+c} g_2^{b+d},\\
\phi_n(g_1)\phi_n(g_2) &  & = &  & \omega^{k+l} g_1^a g_2^b g_1^c g_2^d
&  & = &  & \omega^{k+l+2bc} g_1^{a+c} g_2^{b+d},
\end{array}
\right.
\end{equation}
it follows that $k+l+2ad=2+k+l+2bc$, that is $ad-bc=1$.

$SL(2,\ZZ_3)$ has two generators:
\begin{equation}
S=\left(
\begin{array}{cc}
0 & -1\\
1 &  0
\end{array}\right),\qquad
T=\left(
\begin{array}{cc}
1 & 1\\
0 & 1\\
\end{array}\right).
\end{equation}
We have preimages in $N$ for both $S$ and $T$. Let
\begin{equation}
N_S:=\left(
\begin{array}{ccc}
1 &     1    &   1      \\
1 & \omega^2 & \omega   \\
1 & \omega   & \omega^2 \\
\end{array}\right) \qquad
\text{then}\qquad
N_S g_1 N_S^{-1}=g_2=g_1^0 g_2^1,\qquad
N_S g_2 N_S^{-1}=g_1^2=g_1^{-1} g_2^0,
\end{equation}
hence $N_S\in N$ and $N_S\mapsto S$.\\
Let
\begin{equation}
N_T:=\left(
\begin{array}{ccc}
1 &     0    & 0 \\
0 & \omega   & 0 \\
0 &     0    & 1 \\
\end{array}\right) \qquad
\text{then}\qquad
N_T g_1 N_T^{-1}=g_1=g_1^1 g_2^0,\qquad
N_T g_2 N_T^{-1}=g_1 g_2=g_1^1 g_2^1,
\end{equation}
hence $N_T\in N$ and $N_T\mapsto T$.
\end{proof}

\begin{lemma}
The kernel of $\chi$ is exactly $D\cdot\Gamma$, i.e. any element $n\in \ker\chi$ can be seen as the product of a diagonal matrix times
an element $g\in\Gamma$.
\end{lemma}

\begin{proof}
If $n\in \ker\chi$ then
\begin{equation}
\phi_n(g_1)= n g_1 n^{-1}=\omega^k g_1\qquad \text{and} \qquad
\phi_n(g_2)= n g_2 n^{-1} = \omega^l g_2.
\end{equation}
Let us define
\begin{equation}
g:=g_1^l g_2^{-k}\in\Gamma\qquad \text{and}\qquad n':=ng^{-1}\in N
\qquad \Rightarrow \qquad n=n'g.
\end{equation}
We know that
\begin{equation}
g\, g_1 g^{-1}=(g_1^l g_2^{-k}) g_1 (g_1^l g_2^{-k})^{-1} = \omega^{k} g_1,\qquad
g\, g_2 g^{-1}=(g_1^l g_2^{-k}) g_2 (g_1^l g_2^{-k})^{-1} = \omega^{l} g_2.
\end{equation}
Thus
\begin{eqnarray}
&&\omega^k g_1=n g_1 n^{-1}=n' g\, g_1 g^{-1} {n'}^{-1}=
\omega^k n' g_1 {n'}^{-1}\ ,
\\
&&\omega^l g_2=n g_2 n^{-1}=n' g\, g_2 g^{-1} {n'}^{-1}=
\omega^l n' g_2 {n'}^{-1}\ .
\end{eqnarray}
This means that $n'$ has to commute with $g_1$, therefore is a diagonal matrix, and with $g_2$, which implies that all its entries are equal. Hence $n'\in D$ and the statement follows.
\end{proof}

The map $\chi$ defines a short exact sequence:
\begin{equation}
\label{sequenza1}
0\longrightarrow D\cdot\Gamma\stackrel{i}{\longrightarrow} N
\stackrel{\chi}{\longrightarrow} SL(2,\ZZ_3)\longrightarrow 0,
\end{equation}
where $i$ is the natural inclusion.\\
When mapped to $PGL(3,\CC)$ the sequence \ref{sequenza1} gives a second exact sequence:
\begin{equation}
0\longrightarrow \tilde\Gamma\stackrel{\tilde i}{\longrightarrow}
\tilde N\stackrel{\tilde\chi}{\longrightarrow}
SL(2,\ZZ_3)\longrightarrow 0.
\end{equation}
This proves that $\tilde N/\tilde\Gamma\simeq SL(2,\ZZ_3)$.

For any matrices in $GL(2,\ZZ_3)$ there are $9-1=8$ choices for the elements of the first column and $9-3=6$ choices for the second column, that give $48$ matrices with determinants equal to $\pm 1$. Therefore the group $SL(2,\ZZ_3)$ has cardinality $|SL(2,\ZZ_3)|=24$. Hence the group $\tilde N$ is a finite group of order $9\cdot 24=216$.
\end{proof}

The group $\tilde N$ has the same order of the Hessian group. It is generated by
\begin{equation}
\label{generatori}
\tilde N=<g_1,g_2,N_S,N_T>
\end{equation}
and in the next section we will show that it coincides with the Hessian group.

\subsubsection{The Hesse pencil of plane cubic curves}
The Hesse pencil \cite{Hesse1,Hesse2} is a $1$-parameter family of plane cubic curves $E_\mu\subset\PP^2,\ \mu\in \PP^1$, passing for $9$ base points in particular position. We take the following base points for the pencil and we arrange them in a square array:
\begin{eqnarray}
\begin{array}{ccc}
p_0=(0:1:-1),        & p_1=(1:0:-1),        & p_2=(1:-1:0),\\
p_3=(0:1:-\omega),   & p_4=(1:0:-\omega^2), & p_5=(1:-\omega:0),\\
p_6=(0:1:-\omega^2), & p_7=(1:0:-\omega),   & p_8=(1:-\omega^2:0).
\end{array}
\end{eqnarray}
Thus the pencil is defined by:
\begin{equation}
E_{\mu}: \qquad x_0^3+x_1^3+x_2^3-3\mu x_0 x_1 x_2 =0.
\end{equation}
It is easy to see that there are twelve lines containing the points in the horizontal, vertical and diagonal rows of the above array. Each group of three lines is the support of one of the $4$ singular cubic curves of the pencil:
\begin{eqnarray}
&&\mu=\infty:   \qquad  x_0 x_1 x_2 =0,\cr
&&\mu=1:        \qquad       (x_0+x_1+x_2)(x_0+\omega x_1+\omega^2 x_2)
                        (x_0+\omega^2 x_1+\omega x_2) =0,\\
&&\mu=\omega:   \qquad       (x_0+\omega x_1+x_2)(\omega x_0+x_1+x_2)
                        (x_0+x_1+\omega x_2)=0,\cr
&&\mu=\omega^2: \qquad       (x_0+\omega^2 x_1+x_2)(x_0+x_1+\omega^2 x_2)
                        (\omega^2 x_0+x_1+x_2) =0.\nonumber
\end{eqnarray}
With this choice of the base points, the lines in each singular fiber intersect in three points forming a degenerate orbit of
$\tilde\Gamma$:
\begin{eqnarray}
&&\text{Sing}(E_0)=\text{Fix}(<g_1>)=\{(1:0:0),(0:1:0),(0:0:1)\},\cr
&&\text{Sing}(E_1)=\text{Fix}(<g_2>)=
\{(1:1:1),(1:\omega^2:\omega),(1:\omega:\omega^2)\},\\
&&\text{Sing}(E_\omega)=\text{Fix}(<g_1^2g_2>)=
\{(1:\omega^2:1),(\omega^2:1:1),(1:1:\omega^2)\},\cr
&&\text{Sing}(E_{\omega^2})=\text{Fix}(<g_1g_2>)=
\{(1:\omega:1),(1:1:\omega),(\omega:1:1)\}.\nonumber
\end{eqnarray}
The Hessian group \cite{Jordan,Maschke,1930} $G$ is the subgroup of $PGL(3,\CC)$, the automorphism group of $\PP^2$, that keeps the set of the base points of the Hesse pencil invariant. It sends a curve $E_{\mu}$ of the pencil into another one (eventually itself). This is a group of order $216$. Therefore to show that $\tilde N$, the automorphism group of $\tilde\Gamma$, coincides to the Hesse group it suffices to observe that $\tilde N\subset G$, namely that $\tilde N$ preserves the base points of the pencil. A trivial computation shows that the generators \ref{generatori} of $\tilde N$ satisfy such requirement and this completes the proof of the theorem \ref{theorHessian}.

\subsubsection{The action of $\tilde\Gamma$ on the Hesse pencil}
\label{ActionHessepencil}
The group $\tilde\Gamma\subset\tilde N$ preserves the Hesse pencil structure. Actually the action of $\tilde\Gamma$ on any curve $E_\mu$ of the pencil is the translation by the base points of the pencil. This implies that for any curve $E_\mu/\tilde\Gamma\simeq E_\mu$ and it suggests a possible resolution of $\PP^2/\tilde\Gamma$.

We briefly review how it is defined the group structure on a cubic curve \cite{Hartshorne}. Let $X$ be a non singular cubic curve in $\PP^2$ and let $\text{Pic}^0(X)\subset\text{Pic}(X)$ be the group of degree zero divisors on $X$. A point $P_0$ on $X$ is an inflection point if the intersection multiplicity of the tangent line to $X$ in $P_0$ is equal $3$. A plane cubic curve has exactly $9$ inflection points and any line intersecting any two of them intersects the curve in a third one. The map that to any closed point $P\in X$ associates the divisor $P-P_0\in \text{Pic}^0(X)$ is bijective and gives a group structure on the set of closed points of $X$, the one from $\text{Pic}^0(X)$ with $P_0$ the identity. Similar varieties are known as group varieties. The group law has a nice geometric interpretation. In $\PP^2$ each couple of line $L,L'$ are equivalent in the Picard group $\text{Pic}(\PP^2)$, therefore if $L\cap X={P,Q,R}$ and $L'\cap X={P',Q',R'}$ we have $P+Q+R=P'+Q'+R!
 '$ in $\text{Pic}(X)$. In particular, since $P_0$ is an inflection point, it follows that $P+Q+R=3P_0$ in $\text{Pic}(X)$. This implies that $P+Q+R=0$ in the group variety $X$. The sum of $P$ and $Q$ is equal to the point $T$ such that $P+Q-T=0$, which means that $-T=R\in X\cap L$. However $-T\in X\cap L''$ where $L''$ is the line passing for $T$ and $P_0$. Hence we conclude that the sum of $P$ and $Q$ is the third intersection point $T$ between the curve and the line $L"$ passing for $R$ and the fixed inflection point $P_0$.
\begin{figure}[hbtp]
\centering
\resizebox{7cm}{!}{
\includegraphics[viewport=250 340 360 450,width=5cm,clip]{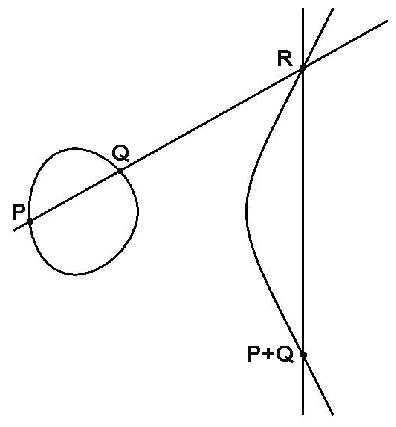}}
\caption{The group law on the cubic curve in $\PP^2$ defined by $Y^2 Z
= X^3-X Z^2$. The inflection point $P_0$ is the point at infinity.}
\label{Fig4} 
\end{figure}

On any smooth cubic curve of the Hesse pencil the $9$ inflections points coincide with the base points of the pencil. It is simple to see that if we fix a point, for example $p_0$, then on any $E_{\mu}$ the sum $p+p_i$ of any point $p\in E_{\mu}$ and the base point $p_i$ is equal to the action of an element of $\tilde\Gamma$ on $p\in\PP^2$ (it is sufficient to prove it for $p$ a base points of the pencil). Actually for any choice of the fixed point $p_i$ we have a group isomorphism $H_i$ between the group of inflection points (the base points of the pencil with group law the one from the group varieties $E_{\mu}$) and $\tilde\Gamma$. For example if we fix $p_0$ then:
\begin{equation}
H_0:(p_0,p_1,p_2,p_3,p_4,p_5,p_6,p_7,p_8)\longmapsto(I,g_2,g_2^2,g_1^2g_2,g_1g_2,g_1g_2^2,g_1^2g_2^2,g_1,g_1^2).
\end{equation}

Thus, the action of $\tilde\Gamma$ on $\PP^2$ is the translation by the base points of the pencil, the point of order three. On any curve $E_\mu$ such translation group is called $E_\mu[3]$. Thus the image $E_\mu/\tilde\Gamma \subset\PP^2/\tilde\Gamma$ is just $E_\mu/E_\mu[3]$ which is well known, being isomorphic to $E_\mu$ under the map:
\begin{eqnarray}
E_\mu   & \stackrel{\cdot 3}{\longrightarrow} & E_\mu/E_\mu[3]\\
\ p     & \longmapsto                         & \quad 3p\nonumber
\end{eqnarray}
If we exclude the base points of the Hesse pencil we can see the projective plane as a bundle of elliptic curves $E_\mu$ on $\PP^1$. We just proved that the quotient map sends any fiber $E_\mu$ of the bundle to an elliptic curve in $\PP^2/\tilde\Gamma$. Thus also $\PP^2/\tilde\Gamma$ contains a natural elliptic pencil, with any fiber isomorphic to one in $\PP^2$. However the projective plane $\PP^2$ is not isomorphic to the singular $\PP^2/\tilde\Gamma$. The pencil in $\PP^2/\tilde\Gamma$ has only one base point, the image of the $9$ base points of the Hesse pencil in $\PP^2$ under the quotient. Therefore the isomorphisms on the fibers are not compatible in the base points of the Hesse pencil and $\PP^2$ is only birational to $\PP^2/\tilde\Gamma$. This suggest that if we blow up $8$ of the $9$ base points of the pencil we should obtain a regular morphism $\Psi$ from a smooth surface $S$, actually a limit of del Pezzo surfaces, to $\PP^2/\tilde\Gamma$ and hence a resolution of such variety. Now we can state the main theorem of the paper:

\begin{theorem}
\label{mainth}
Let $S:=Bl_{p1,\ldots,p8}(\PP^2)$ be the quasi del Pezzo surface obtained as the blow up $\PP^2$ at eight base points of the Hesse pencil. Then we have the following commutative diagram:\\[-1mm]

\hspace{4.7cm}
\xymatrix{
E_\mu \ar@{->}[d]^{\cdot 3} \ar@{^{(}->}[r]
&\PP^2 \ar@{-->}[d]
&S \ar[l]_{\pi} \ar@{->}[ld]^\Psi\\
E_\mu/\tilde\Gamma \ar@{^{(}->}[r] &\PP^2/\tilde\Gamma}\\[3mm]
where $\Psi$ is a desingularization map given by the anticanonical sections on $S$. The total space of the canonical bundle $K_S$ and the map $\Psi$ define a crepant resolution of the orbifold $\CC^3/\Gamma$.
\end{theorem}

\subsection{Del Pezzo surfaces and the resolution of $\PP^2/\tilde\Gamma$}
\label{recalldelPezzo}
In this section we prove theorem \ref{mainth} finding explicitly the map $\Psi$.
\subsubsection{Generalities on del Pezzo surfaces}\label{generalities}
For this section we refer to \cite{Hartshorne,Dolgachev} and to the references given there. A del Pezzo surface is defined to be a surface $X$ with ample anticanonical divisor class $K_X$. This means that it exists a positive integer $n$ such that $n K_X$ is very ample, i.e. the global sections of $n K_X$ give an embedding of $X$ in a projective space. The Nakai-Moishezon criterion says that a divisor $D$ on a surface $X$ is ample if and only if $D^2>0$ and $D\cdot C>0$ for all irreducible curves. Let $S$ be the surface obtained as the blow up of $n$ points $p_1,\ldots,p_n$ in $\PP^2$. The divisor group of $S$ is generated by $H$, the strict transform in $S$ of the form defining the line in $\PP^2$, and the exceptional divisors $E_i$. The intersection pairing is
\begin{equation}
H^2=1,\quad H\cdot E_i=0\quad \text{and}\quad \ E_i\cdot E_j=-\delta_{ij}.
\end{equation}
The canonical divisor of the surface is $-K_S=3H-E_1-\ldots-E_n$. It easily seen that the Nakai-Moishezon criterion implies that $S$ is a del Pezzo surface if $n\leq 8$ and the $p_i$ are in general position, no $3$ of them are collinear and no $6$ of them lie on a conic. A classical result states that every del Pezzo surfaces is either isomorphic to the blow up of $n\leq 8$ points in $\PP^2$ or isomorphic to $\PP^1\times\PP^1$.

{F}rom now on we specialize to the case $n=8$. It is a well known fact that
\begin{equation}
\text{dim}\ H^0(S,-nK_S)=1+\frac{n(n+1)}{2}\qquad (n\geq 1).
\end{equation}
We report the essential properties of the anticanonical maps $\phi_n$, the rational maps defined by the global sections of $-nK_S$, for $n=1,2,3$. We assume that the $8$ points are in general position and that they are all distinct. Thus the points define a pencil of smooth cubic curves in $\PP^2$ passing for them, all intersecting in a ninth points $p_0$. We note that any curve $E$ in the pencil has a strict transform in $S$, again denoted by $E$, linearly equivalent to $-K_S$. Hence for any other $E'$ in the pencil, $E\cdot E'=1$ in $S$ and $p_0$ is the unique points in $S$ where they meet. The anticanonical bundle has $\text{dim}\ H^0(-K_S)=2$, therefore the cubic forms $f_0,f_1$ defining $E$ and $E'$ give a basis of $H^0(-K_S)$. The map $\phi_1$ is defined by:
\begin{eqnarray}
\phi_1: \quad &S& \dashrightarrow \PP^1\\
&q& \mapsto (f_0(q):f_1(q))\nonumber
\end{eqnarray}
This map is not defined in $p_0$ where $f_0(p_0)=f_1(p_0)=0$. Blowing up $p_0$ the map $\phi_1$ defines an elliptic fibration over $\PP^1$, where each fiber is an elliptic curve defined by $\lambda f_0 + \mu f_1=0$ for some $(-\mu:\lambda)\in\PP^1$. Note that the exceptional divisors $E_i$ map onto the $\PP^1$.

Next we have $\text{dim}\ H^0(-2K_S)=4$. The polynomial $f_0^2,f_0 f_1,f_1^2\in H^0(-2K_S)$ therefore it exists a homogeneous polynomial $g$ of degree $6$ such that $H^0(-2K_S)=<f_0^2,f_0f_1,f_1^2,g>$. The map $\phi_2$ is defined by:
\begin{eqnarray}
\phi_2: \quad &S& \dashrightarrow \PP^3\\
&q& \mapsto (X_0:X_1:X_2:X_3)=(f_0(q)^2:f_0(q)f_1(q):f_1(q)^2:g(q))\nonumber
\end{eqnarray}
Recall that $-2K_S=6H-2(E_1+\ldots+E_8)$, thus the sextic curve $g=0$ on $\PP^2$ passes through $p_1,\ldots,p_8$ and it is singular there. Hence the sextic and the cubics $\lambda f_0 + \mu f_1=0$ intersect at $p_1,\ldots,p_8$ with multiplicity $2$. There are two remaining intersection points, but it is not difficult to prove that $p_0$ is not one of them. This means that $g(p_0)\neq 0$ and $\phi_2$ is a morphism. Any fiber of $\phi_1$ is mapped $2:1$ to a $\PP^1$:
\begin{equation}
\phi_2:\ \lambda f_0 + \mu f_1=0\longmapsto
(f_0^2:-\frac{\lambda}{\mu}f_0^2:\frac{\lambda^2}{\mu^2}f_0^2:g)
=(1:-\frac{\lambda}{\mu}:\frac{\lambda^2}{\mu^2}:\frac{g}{f_0^2}).
\end{equation}
Hence $\phi_2$ has degree two onto its image, which is the surface $Q$ in $\PP^3$ of equation $X_0X_2=X_1^2$, a quadric with unique singular point $(0:0:0:1)=\phi_2(p_0)$. The fibers of $\phi_1$ map to the lines passing through the vertex of the cone.

Finally $\text{dim}\ H^0(-3K_S)=7$, thus there is a homogeneous polynomial $h$ of degree $9$ such that $H^0(-3K_S)=<f_0^3,f_0^2f_1,f_0f_1^2,f_1^3,gf_0,gf_1,h>$. The curve $h=0$ in $\PP^2$ has triple points in $p_1,\ldots,p_8$. The map $\phi_3$ is defined by:
\begin{eqnarray}
\phi_3: \quad &S& \dashrightarrow \PP^6\\
&q& \mapsto (f_0^3(q):f_0^2(q)f_1(q):f_0(q)f_1^2(q):f_1^3(q):g(q)f_0(q):g(q)f_1(q):h(q))\nonumber
\end{eqnarray}
It is an embedding which sends each fiber of $\phi_1$ to a smooth cubic in a $\PP^2\subset\PP^6$.

Note that $\text{dim}\ H^0(-6K_S)=22$ but that $H^0(-6K_S)$ contains the $23$ functions $f_0^6,f_0^5f_1,\ldots,f_1^6$, $gf_0^4,\ldots,gf_1^4$, $g^2f_0^2,\ldots,g^2f_1^2$, $g^3$,$hf_0^3,\ldots,hf_1^3$,$h^2$, $f_0gh,f_1gh$. Thus there must be a linear relation among these functions, which is a degree two polynomial in $h$ and its coefficients are polynomials in $f_0,f_1,g$, reflecting the fact that the map defined by $f_0,f_1,g$ is $2:1$. Moreover it can be shown that this relation is unique. Thus the generators $f_0,f_1$ of $H^0(-K_S)$, $g$ of $H^0(-2K_S)$ and $h$ of $H^0(-3K_S)$ define an embedding
\begin{eqnarray}
\label{psiregular}
\Psi: \quad &S& \longrightarrow \PP_{1,1,2,3}\\
&q& \longmapsto (f_0(q):f_1(q):g(q):h(q))\nonumber
\end{eqnarray}
that maps $S$ into a hypersurface of degree six of the weighted projective space.

\subsubsection{The blow up of $\PP^2$ at $8$ base points of the Hesse pencil}
The surface $S$ of theorem \ref{mainth} is defined as the blow up of $\PP^2$ at eight of the base points of the Hesse pencil $p_1,\ldots,p_8$:
\begin{equation}
\pi:\quad S:=Bl_{p1,\ldots,p8}(\PP^2) \longrightarrow \PP^2.
\end{equation}
These points are not in general position because many of them are collinear. Take the strict transform in $S$ of the $8$ lines in the singular fibers of the pencil which do not contain $p_0$.
\begin{eqnarray}
\begin{array}{ccc}
x_1=0:                            & & L_{147}=H-E_1-E_4-E_7, \\
x_2=0:                            & & L_{258}=H-E_2-E_5-E_8, \\
x_0+\omega^2 x_1+\omega x_2=0:    & & L_{345}=H-E_3-E_4-E_5, \\
x_0+\omega x_1+\omega^2 x_2=0:    & & L_{678}=H-E_6-E_7-E_8, \\
x_0+\omega x_1+x_2=0:             & & L_{138}=H-E_1-E_3-E_8, \\
x_0+x_1+\omega x_2=0:             & & L_{246}=H-E_2-E_4-E_6, \\
x_0+\omega^2 x_1+x_2=0:           & & L_{156}=H-E_1-E_5-E_6, \\
x_0+x_1+\omega^2 x_2=0:           & & L_{237}=H-E_2-E_3-E_7.
\end{array}
\end{eqnarray}
The intersection matrix between these curves is given in table \ref{IntPairing}.
\TABLE[h]{
\caption{}\label{IntPairing}
\begin{tabular}{c|cccccccc}
\phantom{x}   &$L_{147}$&$L_{258}$&$L_{345}$&$L_{678}$&$L_{138}$&$L_{246}$&$L_{156}$&$L_{237}$\\
\hline
  $L_{147}$   & -2      & 1       & 0       & 0       & 0       & 0       & 0      & 0       \\
  $L_{258}$   & 1       & -2      & 0       & 0       & 0       & 0       & 0      & 0       \\
  $L_{345}$   & 0       & 0       & -2      & 1       & 0       & 0       & 0      & 0       \\
  $L_{678}$   & 0       & 0       & 1       & -2      & 0       & 0       & 0      & 0       \\
  $L_{138}$   & 0       & 0       & 0       & 0       & -2      & 1       & 0      & 0       \\
  $L_{246}$   & 0       & 0       & 0       & 0       & 1       & -2      & 0      & 0       \\
  $L_{156}$   & 0       & 0       & 0       & 0       & 0       & 0       & -2     & 1       \\
  $L_{237}$   & 0       & 0       & 0       & 0       & 0       & 0       & 1      & -2      \\
\end{tabular}
\label{tab:intpair}
}

As we can see the intersection graph for each pair of curve in the singular fibers is of type $A_2$. Moreover we observe that these curves have zero intersection with the canonical divisor on $S$. Hence $-K_S$ does not satisfy the second requirement of the Nakai-Moishezon criterion and it is not ample. The surface $S$ is a smooth varieties that can be seen as a degenerate limit of del Pezzo surfaces. In particular the anticanonical global sections are constant on the above $L_{ijk}$, therefore such curves get blown down by the anticanonical system of the preceding section. In this case the map (\ref{psiregular}) is a morphism from $S$ to a singular surface of degree six in $\PP_{1,1,2,3}$ which sends the curve $L_{ijk}$ to $4$ singular points of type $A_2$.

\subsubsection{The resolution of $\PP^2/\tilde\Gamma$}
The forms
\begin{eqnarray}
f_0\ &=&\  X Y Z                                                               \cr
f_1\ &=&   X^3 + Y^3 + Z^3                                                     \\
g\   &=&   X^6 + (Y^2 - Y Z + Z^2)^3 + X^3 (2 Y^3 - 3 Y^2 Z - 3 Y Z^2 + 2 Z^3) \cr
h\   &=&  (X^3 + Y^3 + Z^3 + 3 \omega Y^2 Z + 3 \omega^2 Y Z^2)\cdot\cr
     & &  (X^6 + (Y^2 - Y Z + Z^2)^3 + X^3 (2 Y^3 - 3 Y^2 Z - 3 Y Z^2 + 2 Z^3))\nonumber
\end{eqnarray}
define four curves on $\PP^2$. The cubics $f_0=0$ and $f_1=0$ are the ones defining the Hesse pencil, hence they intersect at $p_0,\ldots,p_8$. Their strict transforms in $S$ are divisors linearly equivalent to $-K_S$. The sextic $g=0$ has double points in $p_1,\ldots,p_8$ and it does not contain $p_0$. Therefore it defines a divisor linearly equivalent to $-2K_S$. Finally the curve $h=0$ has triple points at $p_1,\ldots,p_8$ and it does not contain $p_0$. It defines a divisor linearly equivalent to $-3K_S$.

The rational map
\begin{eqnarray}\label{psiregular2}
\overline\Psi: \quad &\PP^2& \dashrightarrow \PP_{1,1,2,3}\\
&q& \mapsto (f_0(q):f_1(q):g(q):h(q))\nonumber
\end{eqnarray}
is not defined in $p_1,\ldots,p_8$, but it gives a morphism
\begin{eqnarray}\label{psiregular3}
\Psi: \quad S \longrightarrow \PP_{1,1,2,3}
\end{eqnarray}
The image of $\Psi$ is the surface of section \ref{c3/d27}:
\begin{equation}
\label{imquasidelP}
w^2+y^3-27wx^3-3wyz+wz^3=0\subset\PP_{1,1,2,3}.
\end{equation}
Note that as it is expected the $(-2)$-curves $L_{ijk}$ on $S$ map to the singular points in \ref{imquasidelP}. We showed that the quotient $\PP^2/\tilde\Gamma$ is isomorphic to \ref{imquasidelP}, therefore we have just proved that such quotient has a desingularization which is a quasi del Pezzo surface $S$. The orbifold $\CC^3/\Gamma$ is isomorphic to the tautological cone over \ref{imquasidelP} in $\CC^4$. It has a desingularization which is the variety $X$ obtained as the total space of the canonical line bundle $K_S$ over $S$, with desingularization map the one from $\Psi$. This complete the proof of theorem \ref{mainth}.

\section{Conclusion}
\label{Conclusion}
In this paper we have studied the geometric properties of the orbifold $X=\CC^3/\Delta_{27}$ and its relations with the geometry of cones over del Pezzo surfaces. The interest in these varieties arised quite recently in the contest of the bottom-up approach to string phenomenology. In a series of paper \cite{Verlinde:2005jr,Buican:2006sn} it has been shown that at low energy an open string theory with a D3-brane placed near the orbifold singularity is identical to the one with a D3-brane near the apex of the cone over del Pezzo surface and that this is a good starting point for a string realization of a Standard Model-like gauge theory.  From a physical point of view this correspondence may be very useful, since, unlike string theory on a general del Pezzo surface, the worldsheet CFT of strings on flat space orbifolds is soluble and the D-brane boundary conditions are exactly known \cite{ALE,Diaconescu:1999dt}.

In our work we use some results contained in \cite{Artebani-Dolgachev} to construct a map between the total space of the canonical bundle of a quasi del Pezzo surface and $X$. This is a morphism that defines a desingularization of the orbifold. As a preliminary step, we have first considered the toric case $Y=\CC^3/{\ZZ_3\times \ZZ_3}$, viewed as a cone over $\PP^2/\ZZ_3$. We have then analyzed a toric resolution of $Y$, realized as the total space of the canonical bundle over a nonsingular quasi del Pezzo surface of degree 3. In this way, the resolution of $Y$ is obtained from the resolution of the quotient space $\PP_2/\ZZ_3$, passing through the (quasi) del Pezzo surfaces in a natural way. We then applied such a viewpoint to the Verlinde-Wijnholt case as summarized in section \ref{sec:intro}.

This work presents many possible developments. Primarily it should provide a useful initial step in the concrete realization of the geometric dual of the minimal quiver extension of the minimal supersymmetric standard model. We recall that in order to obtain this result, Verlinde, Wijnholt and others in \cite{Verlinde:2005jr,Buican:2006sn,Malyshev:2007yb} conceived a cunning symmetry breaking process for the starting quiver theory that gives many indications on the geometric side. In particular it requires the study of monodromies, the research of a partial resolution of the del Pezzo singularity with non isolated $A_2$ singularities and subsequently a particular Calabi-Yau compactification of this local geometry. They present a general prescription to obtain such a goal but the explicit variety, that should determines the actual geometric structure of the hidden dimensions of our world, is still unknown. On such question we observe that, as we showed, on the orbifold side there exist already non isolated $A_2$ singularities. Therefore we suppose that the more natural way to construct the desired partial resolution of del Pezzo singularity should pass through the orbifold geometry and our explicit map between them.

There are several other possible applications. For example we find a natural action of the Heisenberg group on the ``quasi'' del Pezzo geometry side of the theory. This could help to clarify the analysis of \cite{Burrington:2007mj}, where such group has been studied as the automorphisms group of the quiver and in the contest of AdS/CFT correspondence.

However the best known description of D-brane dynamics is in term of derived categories. It should be very interesting to study their properties in this case, also from a mathematical point of view. As a first step one should verify if the geometric correspondence that we find implies the existence of a categorical equivalence suggested by the quiver gauge theories identification. Then one could try to use the knowledge about homological mirror symmetry for del Pezzo surfaces \cite{Orlovealtri} in order to extend the local mirror symmetry conjectures as in \cite{Hosono:2004jp} to the non abelian orbifold case.

\acknowledgments{We are indebted to Bert van Geemen for many explanations and discussions. We are also grateful to Lidia Stoppino and Enrico Schlesinger for very useful discussions and to Stefano Guerra and the anonymous referee for their valuable comments on the original draft of this article. Finally we gratefully acknowledge Giuseppe Berrino, who helped us to improve the exposition of the paper.}

\bibliographystyle{JHEP3}

\end{document}